\documentclass{llncs}[orivec]

\usepackage{comment}

\usepackage{amssymb}%,amsmath}%,aiml06}
\usepackage{latexsym}
\usepackage{amsfonts,amsmath}
\usepackage{alltt}
\usepackage{epic,epsfig}
\usepackage{graphics}
\usepackage{algorithm}
\usepackage{algpseudocode}
\usepackage{algpseudocodex}
\usepackage{algorithm}
\usepackage{xcolor}
\usepackage{float}
\usepackage{tikz}
\usepackage{hyperref} 
\usetikzlibrary{arrows.meta, positioning}

\newcommand{\lgu}{| \hspace{-0.18mm}u \hspace{-0.18mm}|}
\newcommand{\lgw}{| \hspace{-0.18mm}w \hspace{-0.18mm}|}

\newcommand{\lgN}{| \hspace{-0.18mm}N \hspace{-0.18mm}|}
\newcommand{\lgx}[1]{| \hspace{-0.18mm}{#1} \hspace{-0.18mm}|}
\newcommand{\window}{window}
\newcommand{\Window}{Window}

\newcommand{\card}{\mathrm{Card}}
\newcommand{\Ba}{\ensuremath{\Box_a}}
\newcommand{\Bb}{\ensuremath{\Box_b}}
\newcommand{\Bam}{\ensuremath{\Box_a}^{\,\mhyphen}}
\newcommand{\Bbm}{\ensuremath{\Box_b}^{\,\mhyphen}}

\newcommand{\Ra}{{Ra}}
\newcommand{\Rb}{{Rb}}

\newcommand{\Kab}{\ensuremath{\mathbf{\K_{a,b}}}}
\newcommand{\Kdeab}{\ensuremath{\mathbf{KDe_{a,b}}}}
\newcommand{\KQab}{\ensuremath{\mathbf{{\Kab}.4(a)}}}
\newcommand{\KQQab}{\ensuremath{\mathbf{{\Kab}.4(a).4(b)}}}
\newcommand{\KQdeab}{\ensuremath{\mathbf{KDe_{a,b}.4(a)}}}
\newcommand{\KQQdeab}{\ensuremath{\mathbf{KDe_{a,b}.4(a).4(b)}}}

\newcommand{\K}{\ensuremath{\mathbf{K}}}

\newcommand{\pipe}{\hspace{-0.21mm}|\hspace{-0.21mm}}

\newcommand{\inc}{\subseteq}

\newcommand{\upw}{\chi(w)}
\newcommand{\wpp}{\ensuremath{\tilde{w}}}

\newcommand{\Tpp}{\ensuremath{\tilde{T}}}
\newcommand{\nextW}{\mbox{\texttt{NextW}}}
\newcommand{\satW}{\mbox{\texttt{SatW}}}
\newcommand{\chooseW}{\mbox{\texttt{ChooseW}}}
\newcommand{\chooseCCS}{\mbox{\texttt{ChooseCCS}}}
\newcommand{\sat}{\mbox{\texttt{Sat}}}

\newcommand{\CCS}{\mbox{\texttt{CCS}}}
\newcommand{\SF}{\mbox{\texttt{SF}}}
\newcommand{\CSF}{\mbox{\texttt{CSF}}}

\newcommand{\KQsat}{\mbox{\texttt{K4sat}}}

\newcommand{\true}{\mbox{\texttt{True}}}

\newcommand{\all}{\mbox{\texttt{all}}}
\newcommand{\algand}{\mbox{\texttt{and}}}

\usepackage{amsmath}

\def\PSPACE{\mathbf{PSPACE}}
\def\card{\mathtt{Card}}
\def\NPSPACE{\mathbf{NPSPACE}}
\def\NEXPTIME{\mathbf{NEXPTIME}}
\def\coNEXPTIME{\mathbf{coNEXPTIME}}
\def\CPL{\mathbf{CPL}}
\def\K{\mathbf{K}}
\def\At{\mathbf{At}}
\def\Fo{\mathbf{Fo}}
\def\Axiom{\mathbf{A}}
\def\Rule{\mathbf{R}}
\mathchardef\mhyphen="2D
\begin{document}
\title{PSPACE-completeness of bimodal transitive weak-density logic}
\author{Philippe Balbiani \and Olivier Gasquet}
%\footnote{Email: philippe.balbiani@irit.fr.}
%\footnote{Email: olivier.gasquet@irit.fr.}
%
%
\institute{
Institut de recherche en informatique de Toulouse
\\
CNRS-INPT-UT3
}
\maketitle
\begin{abstract}
%With the name of windows, 
{\Window}s have been introduce in \cite{BalGasq25} as a tool for designing polynomial algorithms to check satisfiability of a bimodal logic of weak-density. In this paper, after revisiting the ``folklore'' case of bimodal $\K4$ already treated in \cite{Halpern} but which is worth a fresh review,  we show that {\window}s allow to polynomially solve the satisfiability problem  when adding transitivity to weak-density, by mixing algorithms for bimodal $\K$ together with {\window}s-approach. The conclusion is that both satisfiability and validity are $\PSPACE$-complete for these logics. 
\end{abstract}
%
%
%+
%
\section*{Introduction}

\paragraph{Combining logics} For two normal modal logics $L_a$ and $L_b$, we write $L_a\oplus L_b$ to denote the smallest bimodal logic with two independent modal operators, say $\Box_a$ and $\Box_b$. The complexity of such logics has been addressed in many articles like \cite{Baldoni2,Ladner77,Spaan93}. 
For modal logics defined by \emph{grammar axioms} of the form $\langle a_1\rangle\ldots\langle a_m\rangle p\rightarrow \langle b_1\rangle\ldots\langle b_n\rangle p$, the satisfiability problem is known to be undecidable in general~\cite{FarinasdelCerro1988-FARGL}.
yet, for some specific \emph{grammar logics}, the satisfiability problem is simply known to be decidable like right regular inclusion modal logics of \cite{Baldoni2}, and at the time being, the complexity of $\K+\lozenge p\rightarrow\lozenge\lozenge p$ is only known to be in $\NEXPTIME$. On another hand, some simple grammar logics are known to be in $\PSPACE$, for instance $\K+\lozenge p\leftrightarrow\lozenge\lozenge p$~\cite{FAR-GAS99}. \footnote{We do not mention the word of \cite{Lyon24} as it contains a major and irreparable flaw as discussed in \cite{Gasquet:critique:QD}.}

In this paper, we study the complexity of some modal logics defined by axioms of the form $\langle a\rangle p\rightarrow\langle a\rangle\langle b\rangle p$.
By using a tableau-like approach, we prove that the satisfiability problem of the bimodal logics of transitive weak-density is in $\PSPACE$.

After some basic definitions, we, first, do a little warm-up by revisiting algorithm and complexity of the logic $\KQQab$ (already treated in \cite{Halpern}) in the frame of our settings. Then we briefly review the {\window}s approach presented in \cite{BalGasq25}. %At this time, they were called ``windows'', but we think that 1) ``window'' is a far too connotative and 2) it does not capture the idea of grabbing a whole by looking at small chunks like would do somebody zooming at a picture through a magnifier. 
Then from section \ref{extension:to:transitivity} we transfer results to the transitive cases. 

\section{\Kdeab}

\paragraph{Syntax}
Let $\At$ be the set of all atoms $(p,q,\ldots)$.
The set $\Fo$ of all formulas $(\phi,\psi,\ldots)$ is now defined by
$$\phi:=p\mid\bot\mid\neg\phi\mid(\phi\wedge\phi)\mid\square_{a}\phi\mid\square_{b}\phi$$
where $p$ ranges over $\At$.
As before, we follow the standard rules for omission of the parentheses, we use the standard abbreviations for the Boolean connectives $\top$, $\vee$ and $\rightarrow$ and for all formulas $\phi$, $d(\phi)$ denotes the degree of $\phi$ and $\pipe\phi\pipe$ denotes the number of occurrences of symbols in $\phi$.
For all formulas $\phi$, we write $\lozenge_{a}\phi$ as an abbreviation instead of $\neg\square_{a}\neg\phi$ and we write $\lozenge_{b}\phi$ as an abbreviation instead of $\neg\square_{b}\neg\phi$.
%Let the language \cal{L} of formulas we are concerned with be defined by : 
%\[ \phi :: p \pipe \bot \pipe \neg \phi \pipe \phi \wedge \phi \pipe \Ba \phi \pipe \Bb \phi\]
%NB: we will just write $\Box$ for denoting either modal operators in further definitions. 
%
%
%
%
\paragraph{Semantics}
A {\em frame}\/ is now a $3$-tuple $(W,R_{a},R_{b})$ where $W$ is a nonempty set and $R_{a}$ and $R_{b}$ are binary relations on $W$.
A frame $(W,R_{a},R_{b})$ is {\em weakly dense}\/ if for all $s,t\in W$, if $sR_{a}t$ then there exists $u\in W$ such that $sR_{a}u$ and $uR_{b}t$.
A {\em valuation on a frame $(W,R_{a},R_{b})$}\/ is a function $V\ :\ \At\longrightarrow\wp(W)$.
A {\em model}\/ is a $4$-tuple consisting of the $3$ components of a frame and a valuation on that frame.
A {\em model based on the frame $(W,R_{a},R_{b})$}\/ is a model of the form $(W,R_{a},R_{b},V)$.
With respect to a model $(W,R_{a},R_{b},V)$, for all $s\in W$ and for all formulas $\phi$, the {\em satisfiability of $\phi$ at $s$ in $(W,R_{a},R_{b},V)$}\/ (in symbols $s\models\phi$) is inductively defined as usual.
In particular,
\begin{itemize}
\item $s\models\square_{a}\phi$ if and only if for all $t\in W$, if $sR_{a}t$ then $t\models\phi$,
\item $s\models\square_{b}\phi$ if and only if for all $t\in W$, if $sR_{b}t$ then $t\models\phi$.
\end{itemize}
As a result,
\begin{itemize}
\item $s\models\lozenge_{a}\phi$ if and only if there exists $t\in W$ such that $sR_{a}t$ and $t\models\phi$,
\item $s\models\lozenge_{b}\phi$ if and only if there exists $t\in W$ such that $sR_{b}t$ and $t\models\phi$.
\end{itemize}
A formula $\phi$ is {\em true in a model $(W,R_{a},R_{b},V)$}\/ (in symbols $(W,R_{a},R_{b},V)\models\phi$) if for all $s\in W$, $s\models\phi$.
A formula $\phi$ is {\em valid in a frame $(W,R_{a},R_{b})$}\/ (in symbols $(W,R_{a},R_{b})\models\phi$) if for all models $(W,R_{a},R_{b},V)$ based on $(W,R_{a},R_{b})$, $(W,R_{a},R_{b},V)\models\phi$.
A formula $\phi$ is {\em valid in a class ${\mathcal C}$ of frames}\/ (in symbols ${\mathcal C}\models\phi$) if for all frames $(W,R_{a},R_{b})$ in ${\mathcal C}$, $(W,R_{a},R_{b})\models\phi$.
\paragraph{Axiomatization}
In our language, a {\em bimodal logic}\/ is a set of formulas closed under uniform substitution, containing the standard axioms of $\CPL$, closed under the standard inference rules of $\CPL$, containing the axioms
\begin{description}
\item[$(\Axiom1_{a})$] $\square_{a}p\wedge\square_{a}q\rightarrow\square_{a}(p\wedge q)$,
\item[$(\Axiom2_{a})$] $\square_{a}\top$,
\item[$(\Axiom1_{b})$] $\square_{b}p\wedge\square_{b}q\rightarrow\square_{b}(p\wedge q)$,
\item[$(\Axiom2_{b})$] $\square_{b}\top$,
\end{description}
and closed under the inference rules
\begin{description}
\item[$(\Rule1_{a})$] $\frac{p\rightarrow q}{\square_{a}p\rightarrow\square_{a}q}$,
\item[$(\Rule1_{b})$] $\frac{p\rightarrow q}{\square_{b}p\rightarrow\square_{b}q}$.
\end{description}
Let $4(a)$ (resp.\ $4(b)$) be the formula $\Ba\Ba p\rightarrow\Ba p$ (resp.\ $\Bb\Bb p\rightarrow\Bb p$) and $De_{a,b}$ be $\Ba\Bb p\rightarrow\Ba p$. \\
Let 1) \Kab\ be the least bimodal logic, 2) \KQab\ be the least bimodal logic containing $4(a)$, 3) \KQQab\ be the least bimodal logic containing both $4(a)$ and $4(b)$, 4) \Kdeab\ be the least bimodal logic containing the formula $De_{a,b}$, and 5) \KQdeab\ be the least one containing both $4(a)$ and $De_{a,b}$, and 4) \KQQdeab\ the one containing in addition $4(b)$. \\ 
As is well-known, if $L$ is one of them, $L$ is equal to the set of all formulas $\phi$ such that $\phi$ is valid in the class of all frames which are weakly-dense if $De_{a,b}\in L$, where $\Ra$ (resp.\ $\Rb$) is transitive if $4(a)\in L$ (resp. $4(b)\in L$). This can be proved by using the so-called canonical model construction e.g.\ in \cite{Chellas80}. 
\paragraph{A decision problem}
Let $DP_{a,b}$ be the following decision problem:
\begin{description}
\item[input:] a formula $\phi$,
\item[output:] determine whether $\phi$ is valid in the class of all weakly dense frames.
\end{description}
Using the fact that the coarsest filtration of a weakly dense model is weakly dense, one may readily prove that $DP_{a,b}$ is in $\coNEXPTIME$.
We will prove in Section~\ref{section:complexity:of:KDeab} that $DP_{a,b}$ is in $\PSPACE$.
%
%
%
%

%
%
%
%
%\paragraph{Axiomatics}
%
%
%\Kdeab\ has all axioms and inference rules of $\K$ for both $\Ba$ and for $\Bb$, and the interaction axiom: $\Ba\Bb p\rightarrow \Ba p$. 
%It is the least bi-modal logic containing the formula $\Ba\Bb p\rightarrow \Ba p$. 
%As is well-known, $\Kdeab$ is equal to the set of all formulas $\phi$ such that $\phi$ is valid in the class of all \Kdeab\ frames.
%This can be proved by using the so-called canonical model construction.
%
%
%
%
\section{Basics}
Let $w$ be a finite set of formulas.
We define $d(w)=\max\{d(\phi):\ \phi\in w\}$ and $\pipe w\pipe=\Sigma\{\pipe\phi\pipe:\ \phi\in w\}$.
Moreover, let $\CSF(w)$ be the least set $u$ of formulas such that for all formulas $\phi,\psi$,
\begin{itemize}
\item $w\subseteq u$,
\item if $\phi \wedge \psi\in u$ then $\phi \in u$ and $\psi\in u$,
\item if $\neg (\phi \wedge \psi)\in u$ then $\neg \phi\in u$ and $\neg \psi\in u$,
\item if $\neg \phi\in u$ then $\phi \in u$.
\end{itemize}
In other respect, $\SF(w)$ is the least set $u$ of formulas s. th.\ for all formulas $\phi,\psi$,
\begin{itemize}
\item $w\subseteq u$,
\item if $\phi \wedge \psi\in u$ then $\phi \in u$ and $\psi\in u$,
\item if $\neg (\phi \wedge \psi)\in u$ then $\neg \phi\in u$ and $\neg \psi\in u$,
\item if $\neg \phi\in u$ then $\phi \in u$,
\item if $\square_{a} \phi\in u$ then $\phi \in u$,
\item if $\neg\square_{a} \phi\in u$ then $\neg\phi \in u$,
\item if $\square_{b} \phi\in u$ then $\phi \in u$,
\item if $\neg\square_{b} \phi\in u$ then $\neg\phi \in u$.
\end{itemize}
If $w$ is a set of formulas:
\begin{itemize}
    \item if $4(a)\in L$: $\Bam(w)=\{\phi, \Ba\phi\colon \Ba\phi\in w\}$
    \item if $4(a)\not\in L:\Bam(w)=\{\phi\colon \Ba\phi\in w\}$. Notice that $d(\Box_{a}^{\mhyphen}(w))\leq d(w)-1$. 
\end{itemize}
Similarly for $\Bbm(w)$. %
\\
\\
For all finite sets $u$ of formulas, let $\CCS(u)$ be the set of all finite sets $w$ of formulas such that $u\subseteq w\subseteq \CSF(u)$ and for all formulas $\phi,\psi$,
\begin{itemize}
\item if $\phi \wedge \psi\in w$ then $\phi \in w$ and $\psi\in w$,
\item if $\neg (\phi \wedge \psi)\in w$ then $\neg \phi\in w$ or $\neg \psi\in w$,
\item if $\neg \neg \phi\in w$ then $\phi \in w$,
\item $\bot\not\in w$,
\item if $\neg \phi\in w$ then $\phi\not\in w$.
\end{itemize}
For all finite sets $u$ of formulas, the elements of $\CCS(u)$ are in fact simply unsigned saturated open branches for tableaux of classical propositional logic (see \cite{Smullyan68}).
As a result, for all finite sets $u$ of formulas, an element of $\CCS(u)$ is called a {\em consistent classical saturation (CCS) of $u$.}
As the reader may easily verify, for all finite sets $u, w$ of formulas, if $w\in \CCS(u)$ then $d(u)=d(w)$ and $\CCS(w)=\{w\}$.
Moreover, there exists an integer $c_{0}$ such that for all finite sets $u, w$ of formulas, if $w\in \CCS(u)$ then $\pipe w\pipe\leq c_{0}.\pipe u\pipe$.
\begin{proposition} [Properties of {\CCS}s]\label{prop-CCS}
For all finite sets $u, v, w, w_1, w_2$ of formulas,
\begin{enumerate}
%
%
%\item\label{one} if $w\in \CCS(u)$ then $d(u)=d(w)$,
%
%
%\item\label{eight} if $w\in \CCS(u)$ then $\CCS(w)=\{w\}$,
%
%
%\item\label{two} there exists an integer $c$ such that for all finite sets $u$ of formulas and for all $w\in \CCS(u)$ then $\pipe w\pipe\leq c.\pipe u\pipe$,
%
%
\item\label{One} if $w\in \CCS(u\cup w_1)$ and $w_1\in \CCS(v)$ then $w\in \CCS(u\cup v)$,
\item\label{Two} if $w\in \CCS(u\cup v)$ then it exists $v_1\in \CCS(u)$ and $v_2\in\CCS(v)$ s.th.\ \ $v_1\cup v_2=w$,
\item\label{Three} if $w\in \CCS(u\cup w_1)$ and $w_1$ is a \CCS\ then it  exists $v_2\in \CCS(u)$ s.th.\ $w_1\cup v_2 = w$,
%
%
%\item\label{five} if $\CCS(u)=\emptyset$ then $u$ is (classically) inconsistent,
%
%
\item\label{Four} if $w\in \CCS(u\cup w_1)$ and $w_1\in \CCS(v)$ then $d(w\setminus w_1)\leq d(u)$,
\item\label{Five} if $u\subseteq v$ and $w\in \CCS(v)$ then  $\SF(u)\cap w\in\CCS(u)$,
\item\label{Six} if $u$ is true at a world $x\in W$ of a \Kdeab-model $M=(W,R_{a},R_{b},V)$, then %there exists $w\in\CCS(u)$ such that $w$ is also true at $x$. T
the set $\SF(u)\cap \{\phi\colon M,x\models \phi\}$ is in $\CCS(u)$.
%one such $w$. 
%
%
\end{enumerate}
\end{proposition}
\begin{proof}
%
%
%\ref{one}-\ref{One} are direct properties of classical open branches of tableaux. \\
Item~(\ref{One}) is an immediate consequence of the properties of classical open branches of tableaux.
As for Item~(\ref{Two}), take $v_1=w\cap\CSF(u)$ and $v_2=w\cap\CSF(v)$.
Item~(\ref{Three}) follows from Item~(\ref{Two}).
Concerning Item~(\ref{Four}), if $w\in \CCS(u\cup w_1)$ then by Item~(\ref{Three}), there exists $w_2\colon w_2\in \CCS(u)$ and $w_1\cup w_2 = w$.
Therefore, $w\setminus w_1\subseteq w_2$ and $d(w\setminus w_1)\leq d(w_2)=d(u)$. 
Item \ref{Five} follows by replacing $\in w$ by $\SF(u)\cap w$ in the definition of $\CCS$. 
Finally, about Item~(\ref{Six}), the reader may easily verify it %verify that indeed the set $\CSF(u)\cap \{\phi\colon M,x\models \phi\}$ is in $\CCS(u)$ 
by applying the definition clauses of $\models$. 
\end{proof}

\section{\KQQab}\label{Kab}

 Let $L$ be \KQQab. Because of Prop. \ref{prop-CCS}.\ref{Six}, testing the $L$-satisfiability of a set $u$ of formulas amounts to testing that of a $\CCS$, since $u$ is $L$-satisfiable if and only if there exists a $L$-satisfiable $w\in\CCS(u)$. Hence, given an initial set of formulas $u$ to be tested, we will rather test a non-deterministically chosen set of $\CCS(u)$.\\
 For modal tableaux with transitivity the termination of algorithms are based on the detection of loops in the sequence of \CCS. The seminal algorithm for logic ${\K}4$ of \cite{Ladner77} makes use of a global stack (denoted by $\Sigma$) which memorizes the context in which previous $\Diamond$-formulas has been treated. But this context cannot be the whole \CCS, loops would happen after an exponentially long path. Instead, the sets of ``propagated'' formulas ($\Bam(w)$ and $\Bam(w)$) are considered. \\
In what follows we use built-in functions \algand\ and \all.
The former function lazily implements a logical ``and".
The latter function lazily tests if all members of its list argument are true. 
Essentially, within our setting, Ladner's algorithm may be formulated as follows (the initial call being $\KQsat(\emptyset,\chooseCCS(\{u\}))$ where $\emptyset$ denotes an empty stack):\\

\setlength{\textfloatsep}{0pt}
\setlength{\floatsep}{0pt}
\begin{algorithm}[H]
 \floatname{algorithm}{Function}
\begin{algorithmic}
\Function{$\KQsat$}{$\Sigma,w$}
\State {return}
\State {\hspace{0.87cm}$\bot\not\in w$}
\State {\hspace{0.2cm}$\algand\ \all\{\sat(\Sigma.(u,\neg \psi),\chooseCCS(u))\colon $}
\State {\hspace{3cm}$\neg\Box\psi\in w, u=\Box^{\mhyphen}(w)\cup\{\neg\psi\}, (\Box^{\mhyphen}(w),\neg\psi)\not\in\Sigma\}$}
\EndFunction
\end{algorithmic}
\end{algorithm}

Superficial differences lie on the fact that Ladner's uses signed formulas and is deterministic (it uses explicit for-loops vs.\ non-deterministic choice).\\
His last condition, $(\Box^{\mhyphen}(u),\neg\psi)\not\in\Sigma$, detects loops on a branch of the recursion tree using polynomial space. Let us see how: let $(u_0,\phi_0),(u_1,\phi_1),\cdots,(u_k,\phi_k)$ be the sequence of values taken by the parameters $u$ and $\phi$ in the recursion tree. They can be understood as contexts for the development of a $\Diamond$-formula. Ladner's arguments is based on the fact that $u_i$s are subsets of $\SF(u)$ and most of all they can only grow, i.e.\ for $0\leq i<k: u_i\subseteq u_{i+1}$, but a sequence of identical $u_i$s can only lead to a sequence of $\lgu$ distinct contexts $(u_i,\phi_i)$, after that, the next $u_i$ must be strictly greater. Hence, the total length of the sequence cannot exceed $\lgu.\card(\SF(u))={\cal O}(\lgu^2)$. Thus $\lgx{\Sigma}={\cal O}(\lgu^3)$. Since $\Sigma$ is implemented as a global stack, the overall space needed for a call to $\KQsat$ is still cubic. \\
Then completeness is ensured by building a model where if $(\Box^{\mhyphen}(u_k),\neg\psi_k)\in\Sigma$, say $(\Box^{\mhyphen}(u_k),\neg\psi_k)=(\Box^{\mhyphen}(u_k'),\neg\psi_k')$ for some $0\leq k'<k$, then $(w_k,w_{k'})\in R$ ($w_{k'}$ being the possible world associated with $\Box^{\mhyphen}(u_k')\cup\{\neg\psi_k'\}$).\\
But for our bimodal logic the argument, as is, is not correct since contexts of the sequence are no more increasing. We propose the following algorithm for \KQQab, directly inspired by that of \cite{Ladner77}, which admits a similar bound (same initial call as above): 

\setlength{\textfloatsep}{0pt}
\setlength{\floatsep}{0pt}
\begin{algorithm}[H]
 \floatname{algorithm}{Function}
\begin{algorithmic}
\Function{$\KQsat_{a,b}$}{$\Sigma,w$}
\State {return}
\State {\hspace{0.87cm}$\bot\not\in w$}
\State {\hspace{0.2cm}$\algand\ \all\{\sat(\Sigma.(u,\neg\psi),\chooseCCS(u\cup\{\neg\psi\}))\colon $}
\State {\hspace{2cm}$x\in\{a,b\},\neg\Box_x\psi\in w, u=\Box_x^{\mhyphen}(w), (u,\neg\psi)\not\in\Sigma\}$}
\EndFunction
\end{algorithmic}
\end{algorithm}
Its soundness and completeness proof being embedded in that of the algorithm for $\KQQdeab$, we don't give it here. We just focus on the polynomial termination argument. \\

In the function $\KQsat_{a,b}$ above, the set $u$ will be called a $a$-heir (resp.\ $b$-heir) of $w$ if the $\Diamond$-formulas under concern is $\neg\Box_a\psi_{i-1}$ (resp.\ $\neg\Box_b\psi_{i-1}$). 
\begin{lemma}\label{limit:recursion:depth}
    Let $(u_1,\psi_1,w_1),(u_2,\psi_2,w_2),\cdots,(u_k,\psi_k,w_k)$, the sequence of values taken by the parameters $u$, $\psi$ and $w$ in a branch of the recursion tree. If we consider only the sequence $(u_1,\psi_1),(u_2,\psi_2),\cdots,(u_k,\psi_k)$, its max length for being without repetition between heirs of the same type ($a$ or $b$) is ${\cal O}(\lgu^4)$. \\
\end{lemma}
\begin{proof}
    Firstly, we will need the following \emph{Fact 1}: 
    let us consider $u_{i-1},u_i,u_{i+1}$ such that $u_i$ is an $a$-heir (resp. a $b$-heir) of $u_{i-1}$ and $u_{i+1}$ a $b$-heir (resp. an $a$-heir) of $u_i$ then $d(u_{i+1})<d(u_{i-1})$. \\
    We just treat the first case ($a$-heir then $b$-heir), the other one is similar. Indeed, let $\psi\in u_{i+1}$ then $\psi\in\SF(\Box_b^{\mhyphen}(w_i)\cup\{\neg\psi_i\})$. If $d(\psi)=d(u_i)$ then $\psi=\Box_b \psi'\in w_i$ (otherwise, $\psi\in\SF(\Box_b^{\mhyphen}(w_i)\cup\{\psi_i\})$ and $d(\psi)<d(u_i)\leq d(u_{i-1})$. Hence, $\psi\in\SF(u_i)$, similarly, if $d(\psi)=d(u_{i-1})$ then $\psi=\Box_a \psi''\in w_{i-1}$, leading to a  contradiction since $\psi$ also equals $\psi=\Box_b \psi'$. \\
    Now, w.l.o.g.\ we can suppose that $(u_1,u_2,\cdots,u_k)$ is divided into ``segments'' of only $a$-heirs, followed by only $b$-heirs, then $a$-heirs, and so on, i.e.\ with $k_0=0$:
    $(u_{k_0+1},u_{k_0+2}\cdots,u_{k_1},u_{k_1+1},\cdots,u_{k_2,},u_{k_2+1},\cdots,u_{k_m})$ with $k_m=k$, such that:\\
    if $j\leq 0$ is even (resp.\ odd), then for $l\in[k_j+1..k_j+1[\colon u_{l+1}$ is an $a$-heir of $u_l$ (resp. a $b$-heir). \\
    Accordingly to the Ladner's argument, the length of each sequence of $a$-heirs and of each sequence of $b$-heirs cannot exceed $\lgu^2$, and because of the Fact 1 above, there can be only $d(u)\leq \lgu$ such subsequences. Hence the overall length of it, namely $k$ is bounded by $\lgu^3$. Now if we consider only $a$-heirs (of $b$-heirs), the same holds: there can be only $d(u)\leq \lgu$ subsequences of $a$-heirs, hence the limit for a repetition between heirs of the same type is $2.d(u)\leq 2.\lgu$ and the memory size of the whole sequence $2.\lgu^4$. As a consequence, if none of the $w_i$ is inconsistent, then there exists $1\leq i<j\leq 2.\lgu^4$ such that this branch of the recursion tree may run infinitely without inconsistency (if we remove the loop-test) on\\
    $(u_1,\psi_1,w_1),\cdots,(u_i,\psi_i,w_i),\cdots,(u_{j-1},\psi_{j-1},w_{j-1}),$\\
    \mbox{}\hspace{3cm}$(u_i,\psi_i,w_i),\cdots,(u_{j-1},\psi_{j-1},w_{j-1}),\cdots$\\
    by infinitely repeating the segment $(u_i,\psi_i,w_i),\cdots,(u_{j-1},\psi_{j-1},w_{j-1}),\cdots$ after $(u_{i-1},\psi_{i-1},w_{i-1})$.\\
    Of course, the result holds as well for \Kab\ and for \KQab, since then a sequence of heirs without repetition would still be of length in ${\cal O}(\lgu^2)$. 
    \end{proof}
    Hence, the satisfiability problem for  \KQQab\ is $\PSPACE$-complete\footnote{As already said, this is not new see e.g.\ \cite{Halpern}; please consider this as warming up.}. 

\section{{\Window}s}
For handling weak-density, we introduced the notion of $\window$ in \cite{BalGasq25}, let us have a look back on it. \\
Let $u$ be a finite set of formulas and $w$ be a \CCS\ of $u$.
Let $k\geq d(w)$.
A {\em $k$-{\window} for $w$}\/ (Fig. \ref{{\window}1}) is a sequence $(w_i)_{0\leq i\leq k}$ of sets of formulas (called {\em dense-successors of $w$}) such that
\begin{enumerate}
\item $w_k\in \CCS(\Bam(w))$,
\item for all $0\leq i < k$, $w_i\in \CCS(\Bam(w)\cup \Bbm(w_{i+1}))$.
\item [] (Notice that if $4(b)\in L$ then for all $0\leq i\leq j \leq k$, $\Bbm(w_{j})\subseteq \Bbm(w_i)$)
\end{enumerate}
An {\em $\infty$-{\window} for $w$}\/ is an infinite sequence $(w_i)_{0\leq i}$ of sets of formulas such that
for all $i\geq0$, $w_i\in \CCS(\Bam(w)\cup \Bbm(w_{i+1}))$.\\

%
%
%Notice that for all $\infty$-{\window}s $(w_i)_{0\leq i}$ for $w$ and for all $i\geq0$, $d(w_i)\leq d(w)-1$. 

\noindent 
\framebox{\begin{minipage}{0.95\textwidth}   
\begin{figure}[H]
\centering
\begin{tikzpicture}[->, >=stealth, scale=0.7, transform shape,font=\large]

% Nodes (only these are styled with circles)
\tikzstyle{state}=[minimum size=0.8cm, inner sep=2pt]

\node[state] (w) at (10,3) {$w$};
\node[state] (w0) at (10,0) {$w_0$};
\node[state] (w1) at (8,0) {$w_1$};
\node[state] (w2) at (6,0) {$w_2$};
\node[state] (wd) at (2,0) {$w_{d(w)}$};

\node[draw, dashed, rounded corners, fit=(wd)(w2)(w1)(w0), inner sep=6pt, label=below:{\small }] {};

% Arrows "a" (solid)
\draw (w) -- node[pos=0.7, left] {\small$a$} (wd);
\draw (w) -- node[pos=0.7, left] {\small$a$} (w2);
\draw (w) -- node[pos=0.7, right] {\small$a$} (w1);
\draw (w) -- node[pos=0.7, right] {\small$a$} (w0);

% Arrows "b" (dashed)
\draw[dotted] (wd) -- node[below] {\small$b*$} (w2);
\draw[dashed] (w2) -- node[below] {\small$b$} (w1);
\draw[dashed] (w1) -- node[below] {\small$b$} (w0);

\end{tikzpicture}
\caption{$d(w)$-{\window} for $w$}
\label{{\window}1}
\end{figure}
\end{minipage}
}
%
%
%\begin{definition}[Continuation of $k$-{\window}s for $w$]\label{def-continuation}
\vspace{1cm}\\
%
%
%Let $u$ be a finite set of formulas and $w$ be a \CCS\ of $u$.
%Let $k\geq d(w)$.
Let $T_0=(w_i)_{0\leq i\leq k}$ and $T_1=(\wpp_i)_{1\leq i\leq k+1}$ be two $k$-{\window}s for $w$: {\em $T_1$ is a continuation of $T_0$ for $w$}\/ iff for all $i\in\{1,\ldots,k\}$, $\wpp_i\in \CCS(\Bbm(\wpp_{i+1})\cup w_i)$ (Fig.\ \ref{{\window}}).
%\end{itemize}
%
%
%\end{definition}
%
%
\begin{lemma}[Property of continuations when $4(b)\not\in L$]\label{prop-cont}
\\If $4(b)\not\in L$:\\
Let $u$ be a finite set of formulas and $w$ be a \CCS\ of $u$. Let $k\geq d(w)$.
Let $T_0=(w_i)_{0\leq i\leq k}$ be a $k$-{\window}s for $w$.
If it exists $T_1=(\wpp_i)_{1\leq i\leq k+1}$ which continues $T_0$ for $w$ then $(w_0,\wpp_1,\wpp_2,\cdots,\wpp_{k+1})$ is a $(k+1)$-{\window} for $w$.
%(and as a consequence, $(\wpp_1,\wpp_2,\cdots,\wpp_{k+1})$ is a $k$-{\window} for $w$).
%
%
\end{lemma}
\begin{proof}
First we prove \emph{Fact 2}: for $1\leq i\leq k\colon d(\wpp_i\setminus w_i) \leq d(w)-k+i-1$ by descending induction on $i\in\{1,\ldots,k\}$.\\
Take $i\in\{1,\ldots,k\}$.
Then, either $i=k$, or $i<k$.
In the former case, $\wpp_k\in \CCS(\Bbm(\wpp_{k+1})\cup w_k)$.
Since $w_k\in \CCS(\Bam(w))$ and $\wpp_{k+1}\in \CCS(\Bam(w))$, then $d(w_k)\leq d(w)-1$ and $d(\Bbm(\wpp_{k+1}))\leq d(w)-2$.
Consequently, $d(\wpp_k\setminus w_k) \leq d(w)-1$.
In the latter case, $\wpp_i\in \CCS(\Bbm(\wpp_{i+1})\cup w_i)$; and since $\wpp_{i+1}=\wpp_{i+1}\cup w_{i+1}=(\wpp_{i+1}\setminus w_{i+1})\cup w_{i+1}$, and $\Bbm(A\cup B)=\Bbm(A)\cup \Bbm(B)$, we have $\wpp_i\in \CCS(\Bbm(\wpp_{i+1}\setminus w_{i+1})\cup \Bbm(w_{i+1})\cup w_i)$; but $\Bbm(w_{i+1})\subseteq w_i$, hence $\wpp_i\in \CCS(\Bbm(\wpp_{i+1}\setminus w_{i+1})\cup w_i)$. 
Now, by Prop.\ \ref{prop-CCS}.\ref{Three}: $\exists u\colon u\in\CCS(\Bbm(\wpp_{i+1}\setminus w_{i+1}))$ and $\wpp_i=w_i\cup u$. Thus $\wpp_i\setminus w_i\subseteq u$, and  $d(\wpp_i\setminus w_i)\leq d(u)=d(\Bbm(\wpp_{i+1}\setminus w_{i+1}))\leq d(\wpp_{i+1}\setminus w_{i+1})- 1 \leq d(w)-k+i-1$ (by IH). 

Now we check that $(w_0,\wpp_1,\wpp_2,\cdots,\wpp_{k+1})$ is a $k+1$-{\window} for $w$ by examining the definition of continuations.
Firstly, $\wpp_{k+1}\in \CCS(\Bam(w))$.
Secondly, since $\wpp_k\in \CCS(\Bbm(\wpp_{k+1})\cup w_k)$ and $w_k\in \CCS(\Bam(w))$, then $\wpp_k\in \CCS(\Bbm(\wpp_{k+1})\cup \Bam(w))$.
Thirdly, take $i\in\{1,\ldots,k-1\}$.
Then, $\wpp_i\in \CCS(\Bbm(\wpp_{i+1})\cup w_i)$ and $w_i\in \CCS(\Bam(w)\cup\Bbm(w_{i+1}))$.
Hence by Prop.\ \ref{prop-CCS}.\ref{One}, $\wpp_i\in \CCS(\Bbm(\wpp_{i+1})\cup\Bam(w)\cup\Bbm(w_{i+1}))$.
Since $T_1$ is a continuation of $T_0$, $w_{i+1}\subseteq \wpp_{i+1}$. Then $\Bbm(w_{i+1})\subseteq \Bbm(\wpp_{i+1})$, and $\wpp_i\in \CCS(\Bbm(\wpp_{i+1})\cup\Bam(w))$.
Fourthly, it remains to prove that $w_0\in \CCS(\Bbm(\wpp_1)\cup \Bam(w))$.
By the Fact 2 above, $d(\wpp_1\setminus w_1)\leq d(w)-k\leq 0$, hence if $\Bb \phi\in\wpp_1$ then $\Bb \phi\in w_1$
%Since $w_1\inc \wpp_1$ and $\wpp_1$ does not contain $\Bb$-formulas which are not in $w_1$, then
and thus $\Bbm(\wpp_1)=\Bbm(w_1)$.
Since $w_0\in \CCS(\Bbm(w_1)\cup \Bam(w))$, then $w_0\in \CCS(\Bbm(\wpp_1)\cup \Bam(w))$. 
\end{proof}
\vspace{-0.3cm}

\begin{lemma}[Loops and existence of infinite {\window}]\label{corollary}
Let $u$ be a finite set of formulas and $w$ be a \CCS\ of $u$. Let $\upw=2^{c_{0}.(d(w)+1).\lgw}$. 
%Let $k\geq d(w)$.
\begin{itemize}
    \item Case $4(b)\not\in L$: there exists $(T_i)_{0\leq i\leq \upw}$ be a sequence of $d(w)$-{\window}s for $w$ such that for all $i< \upw$, $T_{i+1}$ is a continuation of $T_i$ for $w$ iff there exists $(\wpp_i)_{i\geq 0}$ an $\infty$-{\window} for $w$.
    \item  Case $4(b)\in L$: there exists $(w_0,w_1)$ a $2$-{\window} for $w$ such that $\Bbm(w_1)\subseteq\wpp_0$ iff then there exists $(\wpp_i)_{i\geq 0}$ an $\infty$-{\window} for $w$.
\end{itemize}
\end{lemma}
\begin{proof}
First case: ($\Leftarrow$) All sets used in $d(w)$-{\window}s for $w$ have their size  bounded by $c_{0}.\lgw$, then there are at most $2^{c_{0}.(d(w)+1).\lgw}$ distinct $d(w)$-{\window}s for $w$.
Hence, there exists integers $h,\delta$ such that $\delta\neq 0$ and $h+\delta \leq 2^{c_{0}.(d(w)+1).\lgw}$ and $T_h=T_{h+\delta}$.
Let $(\Tpp_i)_{0\leq i}$ be the infinite sequence such that for all $i\leq h$, $\Tpp_i=T_i$ and for all $i>h$, $\Tpp_i=T_{h+((i-h)\!\!\!\mod \delta)}$. By construction, for all $i\geq 0$, $\Tpp_{i+1}$ is a continuation of $\Tpp_i$ for $w$.
For all $i\geq 0$, suppose that $\Tpp_i=(w^i_0,\cdots,w^i_{d(w)})$.
For all $i\geq0$, let $\wpp_i=w^i_0$.
As the reader may easily verify, $(\wpp_i)_{i\geq 0}$ is an infinite {\window} for $w$. ($\Rightarrow$) Obviously, for each $i\leq 0$ if we set $T_i=(\wpp_j)_{i\leq j\leq i+d(w)}$, then $(T_i)_{0\leq j\leq 2^{c_{0}.(d(w)+1).\lgw}}$ is the desired finite sequence of {\window}s for $w$ each being a continuation of the previous. \\
Second case: ($\Rightarrow$) immediate by setting $\wpp_0=w_0$ and for all $i\geq 1\colon \wpp_i=w_1$. ($\Leftarrow$) Since all $\wpp_i$ are subset of the finite set $\SF(w)$ there exists $0\leq i\leq j\leq 2^{c_0.\lgw}$ such that $\wpp_i=\wpp_j$. Then let $w_1=\wpp_i$ and $w_0=\wpp_0\cap(\SF(\Bbm(w_1)\cup\Bam(w)))$ and: on the one hand $\Bbm(w_1)\subseteq w_1$, and on another hand a) $w_1\in \CCS(\Bam(w))$ and b) since $\wpp_0\in\CCS(\Bam(w)\cup\Bbm(\wpp_1))$ and $\Bam(w)\cup\Bbm(w_1)=\Bam(w)\cup\Bbm(\wpp_i)\subseteq\Bam(w)\cup\Bbm(\wpp_1)$ then $w_0\in\CCS(\Bam(w)\cup\Bbm(w_1))$ by Prop.\ \ref{prop-CCS}.\ref{Five}. \\
\end{proof}
\section{Algorithm}
We first review the algorithm for $\Kdeab$ of \cite{BalGasq25} which runs as follows (initial call: $\sat(\chooseCCS(u)))$
\setlength{\textfloatsep}{0pt}
\setlength{\floatsep}{0pt}
\vspace{-1cm}
\begin{algorithm}[H]
 \floatname{algorithm}{Function}
\begin{algorithmic}
\caption{Test for \Kdeab-satisfiability of a set $w$: $w$ must be classically consistent and recursively each $\Diamond$-formula must be satisfied as well as all the dense-successors of $w$.}
\Function{\sat}{$w$}:
\State {return}
\State {\hspace{0.87cm}$w\neq \{\bot\}$}
\State {\hspace{0.2cm}\algand\ \all $\{$\sat(\chooseCCS$(\{\neg\phi\}\cup \Bbm(w))\colon \neg\Bb\phi\in w\}$}
\State {\hspace{0.2cm}\algand\ \all $\{\satW(\chooseW(w,\neg\phi),w,\upw)\colon\neg\Ba\phi\in w\}$}
\EndFunction
\end{algorithmic}
\end{algorithm}
\vspace{-1cm}

\begin{algorithm}[H]
\floatname{algorithm}{Function}
\begin{algorithmic}
\caption{Returns $\{\bot\}$ if $x$ is not classically consistent, otherwise returns one classically saturated open branch non-deterministically chosen}
%Computes a \CCS\ of a set $x$ if possible}
\Function{\chooseCCS}{$x$}
\If  {$\CCS(x)\neq \emptyset$}
\State {return one $w\in \CCS(x)$}
\Else 
\State {return $\{\bot\}$}
\EndIf
\EndFunction
\end{algorithmic}
\end{algorithm}
\vspace{-1.5cm}

\begin{algorithm}[H]
\floatname{algorithm}{Function}
\begin{algorithmic}
\caption{Non-deterministically chooses a $d(w)$-{\window} for $w$ if possible (fig.\ \ref{{\window}1})}
\Function{\chooseW}{$w$,$\neg\phi$}
\If {there exists a $d(w)$-{\window} $(w_0,\cdots,w_{d(w)})$ for $w$ such that $\neg \phi \in w_0$}
\State {return $(w_0,\cdots,w_{d(w)})$}
\Else 
\State {return $(\{\bot\},\cdots,\{\bot\})$}
\EndIf
\EndFunction
\end{algorithmic}
\end{algorithm}
\vspace{-1.5cm}

\begin{algorithm}[H]
\floatname{algorithm}{Function}
\begin{algorithmic}
\caption{Tests the satisfiability of each dense-successor of a {\window} for $w$ and recursively for those of its continuation until a repetition happens or a contradiction is detected}
\Function{\satW}{$((w_0,\cdots,w_{d(w)})$,$w$,$N$}:
\If {$N=0$}
\State {return \true}
\Else
\State {return}
\State {\hspace{0.87cm} \sat$(w_0)$}
\State {\hspace{0.2cm} \algand\ \satW(\nextW$((w_0,\cdots,w_{d(w)}),w),w,N-1)$}
\EndIf
\EndFunction
\end{algorithmic}
\end{algorithm}

\begin{algorithm}[H]
\floatname{algorithm}{Function}
\begin{algorithmic}
\caption{Non-deterministically chooses a continuation of a {\window} for $w$ if possible (fig.\ \ref{{\window}})}
\Function{\nextW}{$T_0=(w_0,\cdots,w_{d(w)})$,$w$}
\If {there exists  a continuation $T_1$ of $T_0$ for $w$}
\State {return $T_1$}
\Else 
\State {return $(\{\bot\},\cdots,\{\bot\})$}
\EndIf
\EndFunction
\end{algorithmic}
\end{algorithm}

\noindent \framebox{\begin{minipage}{0.97\textwidth}   
\begin{figure}[H]
\centering
\begin{tikzpicture}[->, >=stealth, scale=0.7, transform shape,font=\large]

% Nodes (only these are styled with circles)
\tikzstyle{state}=[minimum size=0.8cm, inner sep=2pt]

\node[state] (w) at (10,3) {$w$};
\node[state] (w0) at (10,0) {$w_0$};
\node[state] (w1) at (8,0) {$w_1$};
\node[state] (w2) at (6,0) {$w_2$};
\node[state] (wd) at (2,0) {$w_{d(w)}$};
\node[state] (wd1) at (0,0) {$w_{d(w)+1}$};

\node[draw, dashed, rounded corners, fit=(wd1)(wd)(w2)(w1), inner sep=6pt, label=below:{\small Next $d(w)$-{\window} for $w$} once $\sat(w_0)$ has returned \true] {};

% Arrows "a" (solid)
\draw (w) -- node[pos=0.7, left] {\small$a$} (wd1);
\draw (w) -- node[pos=0.7, left] {\small$a$} (wd);
\draw (w) -- node[pos=0.7, left] {\small$a$} (w2);
\draw (w) -- node[pos=0.7, right] {\small$a$} (w1);
\draw (w) -- node[pos=0.7, right] {\small$a$} (w0);

% Arrows "b" (dashed)
\draw[dashed] (wd1) -- node[below] {\small$b$} (wd);
\draw[dotted] (wd) -- node[below] {\small$b*$} (w2);
\draw[dashed] (w2) -- node[below] {\small$b$} (w1);
\draw[dashed] (w1) -- node[below] {\small$b$} (w0);

\end{tikzpicture}
\caption{Results of $\nextW$}
\label{{\window}}
\end{figure}
\end{minipage}
}\mbox{}\\

\section{Analysis}\label{section:complexity:of:KDeab}

Given a \Kdeab-model $M=(W,\Ra,\Rb,v)$ and a set $s$ of formulas, we will write $M,x\models s$ for $\forall \phi\in s\colon M,x\models \phi$. 

\begin{lemma}[Soundness, Lemma 13 of \cite{BalGasq25}]\label{soundness}\\
If $w$ is a \Kdeab-satisfiable (or just satisfiable) \CCS\ then the call \sat(w) returns \true. 
\end{lemma}

\begin{lemma}[Completeness, Lemma 14 of \cite{BalGasq25}]\label{completeness}\\
Given a set $x$ of formulas, if \sat(\chooseCCS(x)) returns \true, then $x$ is \Kdeab-satisfiable. 
\end{lemma}

\begin{figure}[htbp]
\centering
\begin{tikzpicture}[
    bullet/.style={circle, fill=black, inner sep=1.2pt},
    -, >=Stealth, font=\large,scale=0.9]

% Nœuds (positions fixées pour respecter l’ordre des fils)
\node[bullet] (n1) at (0,0) {};
\node (u) at (0.3,0) {$u$};

\node[bullet] (n7) at (3.4,-1.5) {};
\node[bullet] (n71) at (3.4,-1.5) {};
\node[bullet] (n6) at (2.6,-1.5) {};
\node[bullet] (n6.1) at (2.8,-1.5) {};
\node[bullet] (n6.2) at (3,-1.5) {};
\node[bullet] (n5) at (2.2,-1.5) {};
\node[bullet] (n4) at (1.4,-1.5) {};
\node[bullet] (n4.1) at (1.6,-1.5) {};
\node[bullet] (n4.2) at (1.8,-1.5) {};

\node[bullet] (n2) at (-2,-1.5) {};
\node[bullet] (n3) at (0,-1.5) {};

\node[bullet] (n10)  at (0,-3) {};
\node[bullet] (n9)  at (-1,-3) {};
\node[bullet] (n8) at (-2,-3) {};
\node[bullet] (n11) at (0.9,-3) {};
\node[bullet] (n11.1) at (0.2,-3) {};
\node[bullet] (n11.2) at (0.4,-3) {};
\node[bullet] (n12) at (1.4,-3) {};
\node[bullet] (n12.1) at (1.6,-3) {};
\node[bullet] (n12.2) at (1.8,-3) {};
\node[bullet] (n13) at (2.2,-3) {};

\node[bullet] (n14) at (0,-4.5) {};
\node (w) at (0.3,-4.5) {$w$};

\node[bullet] (n17) at (0,-6) {};
\node[bullet] (n17.1) at (0.2,-6) {};
\node[bullet] (n17.2) at (0.4,-6) {};
\node (w0) at (0,-6.5) {$w_0$};

\node[bullet] (n16) at (-1,-6) {};
\node[bullet] (n15) at (-1.8,-6) {};
\node[bullet] (n15.1) at (-1.6,-6) {};
\node[bullet] (n15.2) at (-1.8,-6) {};
\node[bullet] (n18) at (0.9,-6) {};
\node (wk) at (1.2,-6.5) {$w_{d(w)}$};

\node[bullet] (n19) at (2,-6) {};
\node[bullet] (n19.1) at (2.2,-6) {};
\node[bullet] (n19.2) at (2.4,-6) {};
\node[bullet] (n20) at (2.8,-6) {};
\node[bullet] (n21) at (1.6,-6) {};

% Arêtes de 1
\draw[sloped] (n1) to[bend right=20](n2) ;
\foreach \target in {4,4.1,4.2,5,6,6.1,6.2,71} {
  \draw[densely dotted][sloped] (n1) to[bend left=20](n\target);
}
%\draw (n1) -- (n3) node[pos=-0.2,scale=0.7]{$\neg\Bb \phi_0$} ;
\draw (n1)   to[bend left=10](n3) ;

% Arêtes de 3
\foreach \target in {8,9} {
  \draw[sloped] (n3) to[bend right=20] (n\target);
}
\foreach \target in {12,12.1,12.2,13} {
\draw[sloped][densely dotted] (n3) to[bend left=20] (n\target);
}
\draw [dotted](n3)   to[bend left=10](n10) ;
\foreach \target in {11,11.1,11.2} {
\draw [sloped][densely dotted](n3)   to[bend left=10](n\target);
}
% Arêtes de 10
\draw[dashed] (n10) -- (n14) node[midway, left] {};

% Arêtes de 14
\foreach \target in {15,15.1,15.2,16} {
  \draw[densely dotted][sloped] (n14) to[bend right=20](n\target) ;
}
\foreach \target in {19,19.1,19.2,20} {
  \draw[densely dotted][sloped] (n14) to[bend left=20](n\target) ;
}

\draw [densely dotted](n14)   to[bend left=10](n17) ;
\foreach \target in {17.1,17.2,18} {
\draw [densely dotted](n14)   to[bend left=10](n\target);
}

\draw[] (n14) to[bend left=20] (n21);

%\draw[dotted][double] (n10) -- (n11) node[midway, below,scale=0.8] {};
%\draw[dotted][double] (n17)  -- (n18)node[midway, below,scale=0.8] {};

\node[fill=black!10, draw, rounded corners, fit=(n4)(n5), inner sep=2pt, label=below:] {};
\node[fill=black!10, draw, rounded corners, fit=(n6)(n7), inner sep=2pt, label=below:] {};
\node[draw, rounded corners, fit=(n10)(n11), inner sep=2pt, label=below:] {};
\node[fill=black!10, draw, rounded corners, fit=(n12)(n13), inner sep=2pt, label=below:] {};
\node[draw, rounded corners, fit=(n17)(n18), inner sep=2pt, label=below:] {};
\node[fill=black!10, draw, rounded corners, fit=(n19)(n20), inner sep=2pt, label=below:] {};
\node[fill=black!10, draw, rounded corners, fit=(n15)(n16), inner sep=2pt, label=below:] {};

\node[draw, dotted, rounded corners, fit=(n1)(n3)(n10)(n14)(n18), inner sep=5pt, label=below:] {};

\node[bullet] (x14) at (5,-4.5) {};
\node (w) at (5.3,-4.5) {$w$};

\node[bullet] (x171) at (5.2,-6) {};
\node[bullet] (x172) at (5.4,-6) {};
\node[bullet] (x173) at (5.6,-6) {};
\node (wp1) at (5.3,-6.4) {$\wpp_1$};

\node[bullet] (x18) at (6.2,-6) {};
%\node[bullet] (x19) at (5.1,-6) {};
%\node (wk) at (1.1,-6.2) {$\wpp_{d(w)}$};
\node (wpk1) at (6.6,-6.4) {$\wpp_{d(w)+1}$};

%\draw [densely dotted](n14)   to[bend left=10](np17.1) ;
\foreach \target in {171,172,173,18} {
\draw [densely dotted](x14)   to[bend left=10](x\target);
}
%\draw[dotted][double] (n17.3)  -- (n18)node[midway, below,scale=0.8] {};

\node[draw, rounded corners, fit=(x171)(x172)(x173)(x18), inner sep=2pt, label=below:] {};
\node (ww0) at (4.7,-6.4) {$w_0$};
\node (cross0) at (4.8,-6) {$\times$};
\node[bullet] (x17) at (4.8,-6) {};
\draw [densely dotted](x14)  to [bend left=10]node[midway]{$\times$} (x17)  ;
\node[] (f2) at (4.6,-5.3) {};
\node[] (f1) at (2.8,-5.3) {};
\draw[line width=0.5mm,double,->](f1) -- (f2) node[midway, below] {$\nextW$};

\end{tikzpicture}
\caption{A view of the computation tree of $\sat(u)$ when has just been executed a call $\satW((w_0,\cdots,w_{d(w)}),w,\upw)$. Solid lines are $b$-edges, dotted ones are $a$-edges. Small boxes are {\window}s. The big dotted {\window} shows the part stored in memory. On the right, $(\wpp_1,\wpp_2,\cdots,\wpp_{d(w)+1})$ is a  continuation of $(w_0,\cdots,w_{d(w)})$ for $w$, which will be explored once $\sat(w_0)$ will have returned \true\ ($w_0$ can be forgotten).\\ }
\label{memory}
\end{figure}
\begin{lemma}[Lemma 15 of \cite{BalGasq25}]
$\sat(w)$ runs in polynomial space w.r.t.\ $\lgw$.    
\end{lemma} 
Fig.\ \ref{memory} nd the proof are provided in order to enlighten how windows work. 
\begin{proof}
    First, we recall that functions \all\ and \algand\ are lazily evaluated. \\
    Obviously, \chooseCCS\ runs in polynomial space. 
    On another hand, the size of each $d(w)$-{\window} for $w$ is bounded by $d(w).\lgw$, hence by $\lgw^2$ since $d(w)\leq \lgw$. Thus the functions \chooseW\ and \nextW\ run in polynomial space, namely $\mathcal{O}(\lgw^2)$. It is also clear that functions $\sat$ and $\satW$ terminate since their recursion depth is bounded (respectively by $\lgw$ and $\lgN$) as well as their recursion width.
    Among all of these calls, let $\wpp$ be the argument for which $\sat$ has the maximum cost in terms of space, i.e.\ such that $space(\sat(\wpp)$ is maximal.\\ 
    Let $T_0=(w_0,\cdots,w_{d(w)})$ be a $d(w)$-{\window} for $w$. Let us firstly evaluate the cost of $space(\satW(T_0,w,N))$. For $0\leq i<N$, let $T_{i+1}$ be the result of $\nextW(T_i,w)$ (note that $\lgx{T_i}=\lgx{T_0}$). The function $\satW$  keeps its arguments in memory during the call $\sat(w_0)$ and either terminate or forget them and continue, hence: \\
    $space(\satW(T_0,w,N))$ \\
    $\begin{array}{lll}
    \leq& \max \{&\lgx{T_0}+\lgw+\lgN+space(\sat(w_0)),\\
    &&space(\satW(T_1,w,N-1))\}\\
    \leq& \max \{&\lgx{T_0}+\lgw+\lgN+space(\sat(\wpp)),\\
    &&\lgx{T_1}+\lgw+\lgx{N-1}+space(\satW(T_2,w,N-1))\}\\
    \leq& \max \{&\lgx{T_0}+\lgw+\lgN+space(\sat(\wpp)),\\
    &&\lgx{T_0}+\lgw+\lgx{N-1}+space(\sat(\wpp)),\\
    &&\cdots\\
    &&\lgx{T_0}+\lgw+\lgx{0}\}\\
    
     \leq&& \lgx{T_0}+\lgw+\lgN+space(\sat(\wpp))
     \end{array}$
    
    \noindent Since $N\leq \upw$ and $\lgx{T_0}\leq \lgw^2$, $space\satW(T_0,w,N)$ is bounded by $c'.\lgw^2+space(\sat(\wpp))$ for some constant $c'>0$. \\
    Now, concerning the function $\sat$, it also keeps track of its argument in memory during recursion in order to range over its $\Diamond$-formulas. Thus:\\
    $space(\sat(w))$\\
$\begin{array}{ll}
    &  \leq \lgw+\max \{space(\sat(\wpp)),c'.\lgw^2+space(\sat(\wpp))\}\\
    & \leq (c'+1).\lgw^2+space(\sat(\wpp))
    \end{array}$
    
With respect to the size of the arguments (and since $\lgx{\wpp}\leq \lgw$) we are left with a recurrence equation of the form: 
    $space(\lgw)\leq space(\lgw-1)+(c'+1).\lgw^2$ with $space(0)=1$ which yields $space(\sat(\lgw))={\mathcal O}(\lgw^3)$.
\end{proof}
\begin{theorem}
    $DP_{a,b}$ is $\PSPACE$-complete. 
\end{theorem}
\begin{proof}
    On the one hand, $DP_{a,b}$ is $\PSPACE$-hard since it is a conservative extension of $\K$; on the other hand, our function $\sat$ can decide non-deterministically and within polynomial space whether a set of formulas is $\Kdeab$-satisfiable, $\Kdeab$-satisfiability is in $\NPSPACE$, i.e.\ in $\PSPACE$ (by Savitch' theorem. Thus $DP_{a,b}$ is in co-$\PSPACE$ which is equal to $\PSPACE$. 
\end{proof}

    \section{$\Kdeab+$ transitivity}\label{extension:to:transitivity}
    Now we consider logics $L$ among \KQdeab\ and \KQQdeab. Recall that in the case $4(b)\in L$, {\window}s are just $2$-{\window}, so we need to modify functions $\chooseW$ and $\nextW$ in accordance. We need also to modify functions $\sat$ and $\satW$ for dealing with contexts as for \KQab. Function $\nextW$ is unchanged but unused if $4(b)\in L$. Function $\chooseCCS$ is unchanged. Given $u_0$ and $w_0\in\chooseCCS(u_0)$ if it exists (otherwise $u_0$ is unsat), the initial call is $\sat((u_0,\neg\bot),w_0)$. The last context pushed in the stack $\Sigma$ is $last(\Sigma)$. 

\setlength{\textfloatsep}{0pt}
\setlength{\floatsep}{0pt}
\begin{algorithm}[H]
 \floatname{algorithm}{Function}
\begin{algorithmic}
\caption{Test for $L$-satisfiability of a set $w$: $w$ must be classically consistent and recursively each $\Diamond$-formula must be satisfied as well as all the dense-successors of $w$. }
\Function{\sat}{$\Sigma,w$}:
\State {return}
\State {\hspace{0.87cm}$last(\Sigma)\not\in\Sigma$}
\State {\hspace{0.2cm}$\algand\ w\neq \{\bot\}$}
\State {\hspace{0.2cm}$\algand\ \all\{\sat(\Sigma.(u,\neg\psi),\chooseCCS(u\cup\{\neg\psi\}))\colon \neg\Box_b\psi\in w, u=\Box_b^{\mhyphen}(w)$}
%\State {\hspace{3cm}$\neg\Box_b\psi\in w, u=\Box_b^{\mhyphen}(w)\}$}
\State {\hspace{0.2cm}\algand\ \all $\{\satW(\Sigma.(u,\neg\psi),\chooseW(w,\neg\psi),w,\upw)\colon\neg\Ba\psi\in w\}$}
\EndFunction
\end{algorithmic}
\end{algorithm}
\vspace{-1cm}

\begin{algorithm}[H]
\floatname{algorithm}{Function}
\begin{algorithmic}
\caption{Non-deterministically chooses a $2$ or $d(u)$-{\window} for $w$ depending on whether $4(b)\in L$ or not.}
\Function{\chooseW}{$w$,$\neg\psi,B$}
\If{$4(b)\in L$}
\If {$(w_0,w_{1})$ is a $2$-{\window} for $w$ such that $\neg \psi \in w_0$ and $\Bbm(w_1)\subseteq w_0$}
\State {return $(w_0,w_1)$}
\Else 
\State {return $(\{\bot\},\{\bot\})$}
\EndIf
\Else
\If { $(w_0,\cdots,w_{d(u)})$ is a $d(u)$-{\window} for $w$ such that $\neg \psi \in w_0$ and $\Bbm(w_1)\subseteq w_0$}
\State{return $(w_0,\cdots,w_{d(u)})$}
\Else
\State{return $(\{\bot\},\cdots,\{\bot\})$}
\EndIf
\EndIf
\EndFunction
\end{algorithmic}
\end{algorithm}
\vspace{0cm}

\begin{algorithm}[H]
\floatname{algorithm}{Function}
\begin{algorithmic}
\caption{Tests the satisfiability of the dense-successor of a {\window} for $w$ and recursively for those of its continuation until a repetition happens or a contradiction is detected}
\Function{\satW}{$\Sigma$,$((w_0,\cdots,w_k)$,$w$,$N$}: \mbox{$\#k=1$ if $4(b)\in L$, $k=d(u)$ otherwise}
\If {$N=0$}
\State {return \true}
\Else
\If {$4(b)\in L$}
\State {\hspace{0.2cm} return $\sat(\Sigma,w_0)\, \algand\, \sat(\Sigma,w_1)$}
\Else
\State {\hspace{0.2cm} return $\sat(\Sigma,w_0)\,\algand\,\satW(\nextW((w_0,\cdots,w_k),w),w,N-1)$}
\EndIf
\EndIf
\EndFunction
\end{algorithmic}
\end{algorithm}

\begin{algorithm}[H]
\floatname{algorithm}{Function}
\begin{algorithmic}
\caption{Non-deterministically chooses a continuation $T_1=(w'_0,\cdots,w'_{d(w)})$ of $T_0$ (fig.\ \ref{{\window}}) and returns the pair $C$,$T_1$ where $C$ is the context for $w'_0$}
\Function{\nextW}{$T_0=(w_0,\cdots,w_{d(w)})$,$w$}
\If {there exists  a continuation $T_1$ of $T_0$ for $w$}
\State {return $\Bam(w)\cup\Bbm(w_1),T_1$}
\Else 
\State {return $\emptyset,(\{\bot\},\cdots,\{\bot\})$}
\EndIf
\EndFunction
\end{algorithmic}
\end{algorithm}

\noindent The soundness proof transfers almost straightforwardly. 
\begin{lemma}[Soundness]\label{soundness2}\\
If $w$ is a L-satisfiable \CCS\ then the call \sat(w) returns \true. 
\end{lemma}
\begin{proof} 
Since $w$ is L-satisfiable, then $w\neq \{\bot\}$. Hence the result of $\sat(w)$ rely on that of: \\
$\all \{\sat(\chooseCCS(\{\neg\phi\}\cup \Bbm(w))\colon \neg\Bb\phi\in w\}$\\ 
\indent $\algand$\\
    $\all\{\satW(\chooseW(w,\neg\phi,d(w)),\Bam(w),\upw)\colon\neg\phi\in w\}$\\
We proceed by induction on the number $D(w)$ of calls at $\sat$ in the recursion stack: 
\begin{itemize}
    \item Case $d(w)>c.n^4$ (for some $c>0$ given by lemma \ref{limit:recursion:depth}) : then the sets\\
    $\{\sat(\chooseCCS(\{\neg\phi\}\cup \Bbm(w))\colon \neg\Bb\phi\in w\}$ and\\
    $\{\satW(\chooseW(w,\neg\phi,d(w)),\Bam(w),\upw)\colon\neg\phi\in w\}$\\
    are empty. Hence $\sat(w)$ returns \true.
    \item Case $D(w)\geq 1$: for some L-model $M=(W,\Ra,\Rb,v)$ and $x\in W$,
    $M,x\models w$ and
    \begin{enumerate}
        %\item $w\neq \{\bot\}$
        \item since  $M,x\models w$ then for all $\neg\Bb\phi\in w$, $M,x\models \neg\Bb\phi$. Hence for all $\neg\Bb\phi\in w$, there exists $y\in W$ s.th.\ $(x,y)\in\Rb$ and $M,y\models \neg\phi$ and $M,y\models \Bbm(w)$. Thus for all $\neg\Bb\phi\in w$, if $u_0 =\{\neg\phi\}\cup \Bbm(w)$ then $u_0$ is L-satisfiable. Let $w_0=\CSF(u_0)\cap y$, then by Prop.\ \ref{prop-CCS}.\ref{Six} $w_0\in\CCS(u_0)$ and $w_0$ is L-satisfiable too. Thus by IH (since $D(w_0)<D(w)$), for all $\neg\Bb\phi\in w$ there exists $w_0 \in \CCS(u_0)$ such that $\sat(w_0)$ returns \true. Hence $\all\{\sat(\Sigma,\chooseCCS(\{\neg\phi\}\cup \Bbm(w))\colon \neg\Bb\phi\in w\}$ returns \true.
         \item since  $M,x\models w$ then for all $\neg\Ba\phi\in w$, $M,x\models \neg\Ba\phi$. Hence, for all $\neg\Ba\phi\in w$, there exists an infinite sequence $(y_i)_{i\geq 0}$ such that for $0\leq i$:
         \begin{itemize}
             \item $(x,y_i)\in\Ra$ % for $0\leq i\leq \upw$
             \item $(y_{i+1},y_i)\in\Rb$ %for $0\leq i$< \upw$
             \item $M,y_0\models \neg\phi$
             \item $M,y_i\models \Bam(w)$ %for $0\leq i$\leq \upw$
             \item $M,y_i\models \Bbm(y_{i+1})$ %for $0\leq i$< \upw$
         \end{itemize}
Case $4(b)\not\in L$: let 
         \begin{itemize}
             \item $w_{\upw}=\CSF(\Bam(w))\cap y_{\upw}$
             \item $w_i=\CSF(\Bam(w)\cup\Bbm(w_{i+1}))\cap y_i$ for $0\leq i< \upw$
             \item $w_0=\CSF(\{\neg\phi\}\cup \Bam(w)\cup\Bbm(w_{1}))\cap y_0$
         \end{itemize}        
         By Prop.\ \ref{prop-CCS}.\ref{Six}, these $(w_i)_{0\leq i\leq \upw}$ form a sequence of $L$-satisfiable \CCS\ such that:
         %For $1\leq i\leq \upw$, by prop.\ \ref{prop-CCS}.\ref{Four} and since $w_i\inc y_i$ for $0\leq i\leq \upw$ there exists:
         \begin{itemize}
             \item $w_{\upw}\in \CCS(\Bam(w))$
             \item $w_i\in \CCS(\Bam(w)\cup\Bbm(w_{i+1})$ for $1\leq i< \upw$
             \item $w_0\in\CCS(\{\neg\phi\}\cup \Bam(w)\cup\Bbm(w_{1}))$
         \end{itemize}
         Since $D(w_i)<D(w)$ for each $0\leq i\leq \upw$, then, by IH, $\sat(\Sigma',(w_i)$ returns $\true$ for all $0\leq i\leq \upw$. 
         
        Obviously each subsequence $(w_i,\cdots,w_{i+d(w)})$ is a $d(w)$-{\window} for $w$ and  $(w_{i+1},\cdots,w_{i+d(w)+1})$ is a continuation of it.  Thus for each $\neg\Ba\phi\in w$ the call 
        $\satW(\Sigma,\chooseW(w,\neg\phi),\Bam(w),\upw)$
        will reduce to returning:
\[\sat(\Sigma,w_0) \mbox{ \algand\ }\sat(\Sigma_1,w_1) \mbox{ \algand\ } \ldots \mbox{ \algand\ } \sat(\Sigma_{\upw},w_{\upw})\]
         which is \true.\\
Case $4(b)\in L$: similarly with by $w_i=\CSF(\Bam(w)\cup\Bbm(w_{i+1}))\cap y_i$ for $0\leq i$, we obtain an $\infty$-\window\ for $w$, hence by lemma \ref{corollary}, there exists a $2$-\window\ $(w_0,w_1)$ for $w$, and the call reduces to returning: $\sat(\Sigma_0,w_0) \mbox{ \algand\ }\sat(\Sigma_{1}w_{1}$ which is \true\ too. 

\end{enumerate}
\end{itemize}
\end{proof}

For completeness, we proceed by induction on the structure of formulas w.r.t.\ a model explicitly constructed and take into account that $L$-models are not closed under union (since the union of transitive relations is not transitive). \\
%We suppose $4(b)\in L$. 
Let $\sat(\emptyset,w^+)$ with $w^+\in\CCS(u^+)$ be the initial call. Given some set $u$ and a formula $\psi$, we denote $w^{(u,\psi)}$ the \CCS\ chosen by $\chooseCCS(u,\psi)$ (remark 1: it exists by hypothesis), and if $u$ is an $a$-heir, let $(w^{(u,\psi_i)})_{i\leq 0}$ be a $\infty$-{\window} for $w$ (remark 2: it exists by lemma \ref{corollary}). \\
Let $W$ be the set of all occurrences of the values taken by argument $w$ in calls to $\sat(\Sigma,w)$ (we do not use a pointer structure for sets of formulas, each occurrence is distinct from the other even if they contain the same formulas). Let $\Ra'$ and $\Rb'$ be the smallest relations on $W$ defined as follows:\\
for all $w\in W$ and for all $u=\Bbm(w)$ and $\Bb\neg\psi\in w$, 
\begin{itemize}
    \item if $(u,\psi)\not\in\Sigma$ then $(w,w^{(u,\psi)})\in\Rb'$
    \item else, let $(u,\psi)=(u',\psi')\in\Sigma$ and $u'$ of the same type as $u$, then $(w,w^{(u',\psi)})\in\Rb'$ (backward loop; note that $w^{(u,\psi)}\not\in W$)
\end{itemize} 
and for all $u=\Bam(w)$ and $\Ba\neg\psi\in w$
\begin{itemize}
    \item if $(u,\psi)\not\in\Sigma$ then $(w,w^{(u,\psi)})\in\Ra'$,  and in addition for $i\geq 0\colon (w,w_i^{(u,\psi)})\in\Rb'$ and $(w_{i+1}^{(u,\psi)},w_i^{(u,\psi)})\in\Ra'$
    \item else, let $(u,\psi)=(u',\psi')\in\Sigma$ and $u'$ of  same type as $u$, then $(w,w^{(u',\psi')})\in\Ra'$  and in addition for $i\geq 0\colon (w,w_i^{(u',\psi')})\in\Ra'$; note that $(w_{i+1}^{(u',\psi)},w_i^{(u',\psi')})$ is already in $\Ra'$.\\
\end{itemize} 
Finally, let $\Ra=(\Ra')^+$ (the transitive closure of $\Ra'$), and if $4(b)\in L\colon \Rb=(\Rb')^+$, and $V(p)=\bigcup_{w\in W} V_w(p)$ for all $p\in\At$. \\
\begin{lemma}
    Let $M=(W,\Ra,\Rb,V)$ as defined above, $M$ is an $L$-model. 
\end{lemma}
\begin{proof}
    We just have to check the weak-density condition since, by construction, $\Ra$ and $\Rb$ are transitive. Let $(w_i,w_j)\in\Ra$, i.e.\ $\in(\Ra')^+$ then either $(w_i,w_j)\in\Ra'$ or for some $w_{j-1}\colon (w_i,w_{j-1})\in\Ra$ and $(w_{j-1},w_j)\in\Ra'$. In both cases, $w_j$ is $w_k^{(u_i,\psi_i)}$, a member of a $\infty$-{\window} for $w_i$ (or for $w_{j_1}$), hence still by construction there exists $w_{k+1}^{(u_i,\psi_i)}$ such that $(w_i,w_{k+1}^{(u_i,\psi_i)})\in\Ra$ (or $(w_{j-1},w_{k+1}^{(u_i,\psi_i)})\in\Ra$) and $(w_{k+1}^{(u_i,\psi_i)},w_k^{(u_i,\psi_i)})\in\Rb$. We are done in the first case, and in the second one, we conclude by observing that $(w_i,w_{k+1}^{(u_i,\psi_i)})$ by transitivity, and we are done too.   \\
\end{proof}
\begin{lemma}
    For any $w=\CCS(u)\in W$ for some $u$, $\phi\in w$ iff $M,w\models\phi$. 
\end{lemma}
\begin{proof} By induction of the structure of $\phi$, we only treat the modal and atomic cases. 
\begin{itemize}
    \item if $\phi=p\in\At$ for some $p$, then $p\in w$ and by definition $w\in V_w(p)$, hence $M,w\models \phi$
    \item if $\phi=\neg\Bb\psi$ then with $u=\Bbm(w)$, we have $w^{(u,\neg\phi)}\in W$ (cf.\ remark 1), $(w,w^{(u,\neg\phi)})\in\Rb$, and $\neg\phi\in w^{(u,\neg\phi)}$, hence by IH $M,w^{(u,\neg\phi)}\models \neg\phi$, hence $M,w\models \phi$
    \item if $\phi=\neg\Ba\psi$ then with $u=\Bbm(w)$ and $(w_k^{(u,\neg\phi)})_{k\leq 0}$ a $\infty$-{\window}  for $w$ (cf.\ remark 2) such that $\neg\psi\in w_0^{(u,\neg\phi)}$, we have $w_0^{(u,\neg\phi)}\in W$, $(w,w_0^{(u,\neg\phi)})\in\Rb$, and $\neg\phi\in w_0^{(u,\neg\phi)}$, hence by IH $M,w_0^{(u,\neg\phi)}\models \neg\phi$, hence $M,w\models \phi$
    \item if $\phi=\Ba\psi$, let $(w,w')\in\Ra$, the reader may verify that in any case, $\Bam(w)\subseteq w'$, thus $\psi\in w'$ and by IH, $M,w'\models \psi$, hence $M,w\models \phi$. Similarly for $\Bb\psi$. 
    \end{itemize}
\end{proof}
The attentive reader will notice that since $\infty$-{\window}s when $4(b)\in L$ are of the form $(w_0,w_1,w_1,\cdots)$, then in fact instead of constructing a model with infinitely identical copies of the same world, we could as well have added only one copy of $w_1$ and a unique reflexive edge $(w_1,w_1)$ to $\Rb$. The model would be simpler but the proof a little more complex, with no change in the complexity. 

\begin{lemma}
    $Sat(\Sigma,w)$ runs in polynomial space w.r.t.\ $\lgw$. 
\end{lemma}
\begin{proof}
    About the complexity of $\satW(W,w,B)$ the reasoning is exactly the same as for $\Kdeab$ (and even simpler if $4(b)\in L$). \\
        Now, concerning calls $\sat(\Sigma,w)$, as already said, $\Sigma$ is implemented as a global stack of size $\lgw^4$, and the maximal recursion depth is $2.\lgw^3$. But it still keeps track of its argument in memory during recursion in order to range over its $\Diamond$-formulas. Thus, if we omit $\Sigma$:\\
    $space(\sat(\Sigma,w))$\\
$\begin{array}{ll}
    &  \leq \lgw+\max \{space(\sat(\wpp)),c'.\lgw^2+space(\sat(\wpp))\}\\
    & \leq (c'+1).\lgw^2+space(\sat(\wpp))\\
    \end{array}$

\noindent Here, $\lgx{\wpp} $ is no more smaller than $\lgw$, but we know the recursion depth is $2.\lgw^3$, hence we have: $space(\Sigma,w)=\lgx{\Sigma}+space'(\Sigma,w,2.\lgw^3)$ with: 
\begin{itemize}
    \item $space'(\Sigma,w,0)=1$
    \item $space'(\Sigma,w,n)=(c'+1).\lgw^2+space'(\Sigma,w,n-1)$
\end{itemize}
 which yields $space(\sat(\Sigma,w))=\lgx{\Sigma}+{\cal O}(\lgw^3)\times {\cal O}(\lgw^2)={\cal O}(\lgw^5)$. 
\end{proof}

\begin{theorem}
    For $L \in\{\Kdeab+4(a), \Kdeab+4(b),\Kdeab+4(a)+4(b)\}$, $DP_{L}$ is $\PSPACE$-complete. 
\end{theorem}
\begin{proof}
    All these logics are $\PSPACE$-hard since they are all conservative extensions of $K$ or of $K4$, and on the other hand, function $\sat$ can decide non-deterministically and within polynomial space whether a set of formulas is satisfiable, hence satisfiability is in $\NPSPACE$, i.e.\ in $\PSPACE$ (by Savitch' theorem), and so $DP_{L}$ is in co-$\PSPACE$, i.e.\ in $\PSPACE$. 
\end{proof}
\section*{Conclusion}
After having successfully been applied to weak-density alone in \cite{BalGasq25}, the {\window}s approach proves to be useful beyond this case. One may ask whether there is a connection of our {\window}s with so-called mosaics of \cite{Marx00} that were first introduced in \cite{Nemeti95}. In fact, even if {\window}s may be viewed as a kind of overlapping mosaics, membership in $\PSPACE$ is mostly due to this overlapping which is the important feature. Thus, the answer seems rather to be ``yes but''. {\Window}s proves here to be the adequate tool for polynomially examine structures that can serve to build a model. We should be applicable to more open questions of complexity/decidability for logics having similar properties by defining more complex \window\ structures.

\end{document}